\documentclass[sigplan,nonacm,10pt]{acmart}

\usepackage{hyperref}
\usepackage{cleveref}
\crefname{algocf}{alg.}{algs.}
\Crefname{algocf}{Algorithm}{Algorithms}

\usepackage{enumerate}

\RequirePackage[noline,noend,linesnumbered]{algorithm2e}

\SetInd{0.5em}{0.5em}
\DontPrintSemicolon
\SetKwBlock{cstate}{per-replica CRDT state:}{}
\SetKwBlock{query}{query}{}
\SetKwBlock{update}{update}{}
\SetKwBlock{generator}{generator}{}
\SetKwBlock{effector}{effector}{}
\SetKwBlock{crun}{run}{}
\SetKwBlock{kwfunction}{function}{}

\theoremstyle{definition}

\newcommand{\mc}[1]{\ensuremath{\mathcal{#1}}}
\newcommand{\msf}[1]{\ensuremath{\mathsf{#1}}}

\newcommand{\mb}[1]{\ensuremath{\mathbb{#1}}}
\newcommand{\ra}{\rightarrow}
\newcommand{\N}{\mb{N}}

\begin{document}

\title{For-Each Operations in Collaborative Apps}

\author{Matthew Weidner}
\email{maweidne@andrew.cmu.edu}
\orcid{0000-0003-0701-7676}

\affiliation{
    \institution{Carnegie Mellon University}
    \streetaddress{5000 Forbes Ave}
    \city{Pittsburgh}
    \state{Pennsylvania}
    \country{USA}
    \postcode{15213}
}

\author{Ria Pradeep}
\email{rpradeep@alumni.cmu.edu}

\affiliation{
    \institution{Carnegie Mellon University}
    \streetaddress{5000 Forbes Ave}
    \city{Pittsburgh}
    \state{Pennsylvania}
    \country{USA}
    \postcode{15213}
}

\author{Benito Geordie}
\email{bg31@rice.edu}
\orcid{0000-0002-4021-0016}
\affiliation{
    \institution{Rice University}
    \streetaddress{6100 Main St}
    \city{Houston}
    \state{Texas}
    \country{USA}
    \postcode{77005}
}

\author{Heather Miller}
\email{heather.miller@cs.cmu.edu}
\orcid{0000-0002-2059-5406}

\affiliation{
    \institution{Carnegie Mellon University}
    \streetaddress{5000 Forbes Ave}
    \city{Pittsburgh}
    \state{Pennsylvania}
    \country{USA}
    \postcode{15213}
}


\begin{abstract}
Conflict-free Replicated Data Types (CRDTs) allow collaborative access to an app's data. We describe a novel CRDT operation, for-each on the list of CRDTs, and demonstrate its use in collaborative apps. Our for-each operation applies a given mutation to each element of a list, including elements inserted concurrently. This often preserves user intention in a way that would otherwise require custom CRDT algorithms. We give example applications of our for-each operation to collaborative rich-text, recipe, and slideshow editors.
\end{abstract}

\begin{CCSXML}
<ccs2012>
   <concept>
       <concept_id>10003752.10003809.10010172</concept_id>
       <concept_desc>Theory of computation~Distributed algorithms</concept_desc>
       <concept_significance>500</concept_significance>
       </concept>
   <concept>
       <concept_id>10003120.10003130.10003233</concept_id>
       <concept_desc>Human-centered computing~Collaborative and social computing systems and tools</concept_desc>
       <concept_significance>500</concept_significance>
       </concept>
 </ccs2012>
\end{CCSXML}

\ccsdesc[500]{Theory of computation~Distributed algorithms}
\ccsdesc[500]{Human-centered computing~Collaborative and social computing systems and tools}

\keywords{collaboration, CRDTs, concurrency}

% TODO: rights, DOI, conference info from rights form

\maketitle

\section{Introduction}

Lists of mutable values are common in collaborative apps. Examples include the list of slides in a slideshow editor, or the list of rich characters (characters plus formatting attributes) in a rich-text editor.

To allow collaborative access to this data, we would like to use a Conflict-free Replicated Data Type (CRDT) \cite{Shapiro:2011, crdt_summary_2018}. Let us assume that we already have a CRDT $\mc{C}$ representing the list's mutable value type. Then one can construct a \emph{list of $\mc{C}$s} CRDT representing the entire list.

The basic operations on a list of $\mc{C}$s allow users to insert, delete, and apply $\mc{C}$ operations to individual list elements. However, many user operations instead take the form of a ``for-each'' loop: for each element of the list meeting some condition, apply a $\mc{C}$ operation to that element. For example:
\begin{itemize}
    \item In a slideshow editor, a user selects slides 3--7 and changes the background color to blue. This does: for each slide, if its index is in the range $[3, 7]$, then set its "background color" property to "blue".
    \item In a rich-text editor, a user selects all and clicks the "bold" formatting button. This does: for each rich character, set its "bold" formatting attribute to "true".
\end{itemize}

The easy way to implement such for-each operations is using a literal for-each loop on the initiating user's replica. That is, the user's device loops through its replica of the list and performs CRDT operations on individual elements. We call this a \emph{$\msf{forEachPrior}$} operation, since it acts on elements that were inserted (causally) prior to the for-each operation.

However, $\msf{forEachPrior}$ operations do not always capture user intention. \Cref{fig:rich_text_middle} shows a classic example: in a rich text editor, if one user bolds a range of text, while concurrently, another user types in the middle of the range, then the latter text should also be bolded. A literal for-each loop on the first user's replica will not do so because it is not aware of concurrently-inserted characters.

\begin{figure}
    \centering
    \includegraphics[width=0.47\textwidth]{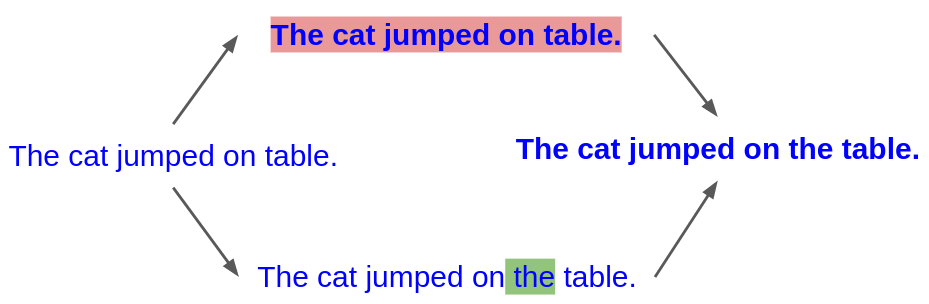}
    \caption{Typical user intention for a bold-range operation (top) concurrent to text insertion (bottom). When using a $\msf{forEachPrior}$ operation, " the" would not be bolded.}
    \label{fig:rich_text_middle}
\end{figure}

Traditionally, collaborative rich-text editors accomplish \Cref{fig:rich_text_middle}'s user intention using a specialized tree structure (e.g.\ Ignet et~al.\ \cite{ignet_rich_text}) or formatting markers at both ends of the range (e.g.\ Peritext \cite{peritext}). However, both techniques require a careful analysis of operations' interactions, and they do not generalize beyond rich-text editing.

% Long version: more detail about difficulties/shortcomings; likewise for similar attempts at for-each: XML, trees?

\subsection{Contributions}
In this paper, we propose a novel for-each operation on lists of CRDTs. Unlike $\msf{forEachPrior}$, it applies a given mutation to every list element that is inserted prior \emph{or concurrently} to the for-each operation. We call this operation \emph{for-each} (without qualification), to distinguish it from $\msf{forEachPrior}$.

\Cref{fig:regions} illustrates how an element's insert operation may relate to a for-each operation: causally prior, concurrent, or causally future. Our for-each operation affects the ``prior'' and ``concurrent'' categories, while $\msf{forEachPrior}$ only affects the ``prior'' category.

\begin{figure}
    \centering
    \includegraphics[width=0.4\textwidth]{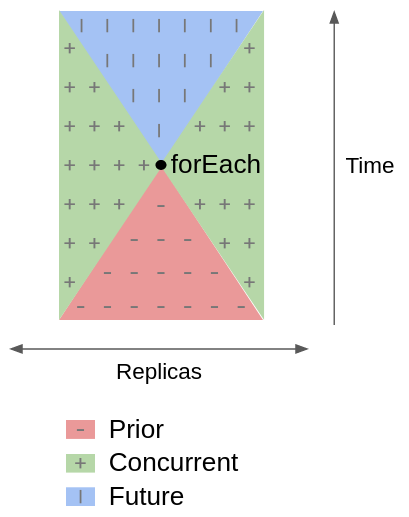}
    \caption{Light cone diagram for a $\msf{forEach}$ operation.}
    \label{fig:regions}
\end{figure}

Using our for-each operation, we easily implement \Cref{fig:rich_text_middle}'s intended behavior: issue a for-each operation with the mutation ``if the character is in the range, set its "bold" formatting attribute to "true"''. See \Cref{sec:ex_rich_text} for details.

We hope that the intended semantics (i.e., user-visible behavior) of the list of $\mc{C}$s and our for-each operation are already clear. However, some technicalities arise, especially when applying for-each to concurrently-inserted elements. Sections \ref{sec:list_of_crdts} and \ref{sec:for_each} discuss these technical details, including algorithms for all of our constructions.

A hurried reader may skip directly to \Cref{sec:examples}, which applies for-each operations to example collaborative apps.

\subsection{Background}
We assume familiarity with the causal order on CRDT operations \cite{causal_order}. The terms ``(causally) prior'', ``concurrent'', and ``(causally) future'' reference this order. Our algorithms use vector clocks \cite{fidge, mattern} to query the causal order relationship between CRDT operations.

Throughout the paper, we use the language of operation-based CRDTs \cite{Shapiro:2011}, although our constructions can easily be reformulated as state-based CRDTs. Each CRDT operation is described in terms of a \emph{generator} and an \emph{effector}. The generator is called to handle user input on the user's local replica, and it returns a \emph{message} to be broadcast to other replicas. Each replica, including the sender, applies the operation by passing this message to the corresponding effector; the sender does so atomically with the generator call. We assume that messages are received exactly once on each replica, and in causal order.

% Long-version: state-based alg, as mentioned above

\section{List of CRDTs}
\label{sec:list_of_crdts}
We begin with a formal description of the \emph{list of CRDTs}. It is modeled on Yjs's \texttt{Y.Array} shared type \cite{yjs}.

First, a \emph{list CRDT} is a classic CRDT type whose external interface is a list (ordered sequence) of immutable values, e.g., the characters in a text document \cite{Shapiro:2011, attiya}. Since the same value may appear multiple times in a list, we use \emph{element} to refer to a unique instance of a value. A list CRDT has operations to insert and delete elements.

A \emph{list of CRDTs} is a more general CRDT in which the list values are themselves mutable CRDTs. Specifically, let $\mc{C}$ be an operation-based CRDT. The external interface of a \emph{list of $\mc{C}$s} is a list of mutable values of type $\mc{C}$. The operations on the list are:
\begin{itemize}
    \item $\msf{insert}(i, \sigma)$: Inserts a new element with initial value $\sigma$---a state of $\mc{C}$---into the list at index $i$, between the existing elements at indices $i-1$ and $i$. All later elements (index $\ge i$) shift to an incremented index.
    \item $\msf{delete}(i)$: Deletes the element at index $i$. All later elements (index $\ge i+1$) shift to a decremented index.
    \item $\msf{apply}(i, o)$: Applies a $\mc{C}$ operation $o$ to the element at index $i$. All replicas update their copy of the element's value (a state of $\mc{C}$) in the usual way for $\mc{C}$ operations. A concurrent delete operation may cause a replica to receive the apply message after deleting the element; in this case, it ignores the apply message (the delete ``wins'').
    %: the initiating replica passes $o$ and the element's current value (a state of $\mc{C}$) to a \emph{generator}, which returns a message to send; this message is broadcast exactly-once in causal order to all replicas; and all replicas pass the message and the element's current value to an \emph{effector}, which outputs its new value.
\end{itemize}

\begin{example}
\label{ex:rich_text_list}
As a running example, consider a collaborative rich-text document, such as a Google Doc. We can represent one character using a \emph{rich character CRDT}. Its state is a pair $(\mathit{char}, \mathit{attrs})$, where $\textit{char}$ is an immutable character and $\mathit{attrs}$ is a map CRDT \cite{Shapiro:2011} for formatting attributes. Then a \emph{list of rich character CRDTs} models the entire rich-text document's state. E.g., the text ``a\textbf{b}'' is represented as
\begin{align*}
[&\{\mathit{char}: \text{\texttt{"a"}}, \mathit{attrs}: \{\}\},\\
&\{\mathit{char}: \text{\texttt{"b"}}, \mathit{attrs}: \{\text{\texttt{"bold"}}, \msf{true}\}\}]
\end{align*}
\end{example}

We construct the list of $\mc{C}$s using $\mc{C}$ and an ordinary list CRDT $\mc{L}$. See \Cref{alg:list_of_crdts} for pseudocode.

Specifically, we assume that $\mc{L}$ produces \emph{positions} that are unique, immutable, and drawn from a dense total order $<$, e.g., Logoot's ``position identifiers'' \cite{logoot}.\footnote{If $\mc{L}$ uses extra state (e.g., tombstones) or messages to manage its positions, then those are implicitly added to the state or messages for the list of $\mc{C}$s.} Then the list of $\mc{C}$s is implemented as:
\begin{description}
    \item[State] A list of elements $(p, \sigma)$, where $p$ is a position from $\mc{L}$ and $\sigma$ is a state of $\mc{C}$, sorted by $p$. An application using the list usually only looks at the values $\sigma$, but it may also use the positions, e.g., for cursor locations.
    \item[Insert, delete] Similar to $\mc{L}$.
    \item[Apply] Similar to $\mc{C}$, except that the message sent to remote replicas is tagged with the element's position $p$. In the pseudocode, we use $\mc{C}.\msf{gen}(o, \sigma)$ to represent $\mc{C}$'s generator for an operation $o$, and we use $\mc{C}.\msf{eff}(m, \sigma)$ to represent $\mc{C}$'s effector for a message $m$.
\end{description}

\begin{algorithm}[h]
    \cstate{
        $\mathit{elts}$: A list of elements $(p, \sigma)$, where $p$ is a position from $\mc{L}$ and $\sigma$ is a state of $\mc{C}$, sorted by $p$ \;
    }
    \BlankLine
    
    \query({$\msf{elements}()$}){
        \Return{$elts$}
    }
    \BlankLine
    
    \update({$\msf{insert}$}){
        \generator({$(i, \sigma)$}){
            $p \gets$ new $\mc{L}$ position between the positions at indices $i-1$ and $i$ in $\mathit{elts}$ \;
            \Return{$(\msf{insert}, p, \sigma)$}
        }
        \BlankLine

        \effector({$(\msf{insert}, p, \sigma)$}){
            Insert $(p, \sigma)$ into $\mathit{elts}$
        }
    }
    \BlankLine

    \update({$\msf{delete}$}){
        \generator({$(i)$}){
            $\mathit{elt} \gets$ $i$-th element in $elts$ \;
            \Return{$(\msf{delete}, \mathit{elt}.p)$}
        }
        \BlankLine

        \effector({$(\msf{delete}, p)$}){
            $\mathit{elt} \gets$ unique element of $\mathit{elts}$ s.t.\ $\mathit{elt}.p = p$, or $\msf{null}$ if none exists (already deleted) \;
            \If{$\mathit{elt} \neq \msf{null}$}{
                Delete $\mathit{elt}$ from $\mathit{elts}$
            }
        }
    }
    \BlankLine

    \update({$\msf{apply}$}){
        \generator({$(i, o)$}){
            $\mathit{elt} \gets$ $i$-th element in $elts$ \;
            $m \gets \mc{C}.\msf{gen}(o, \mathit{elt}.\sigma)$ \;
            \Return{$(\msf{apply}, \mathit{elt}.p, m)$}
        }
        \BlankLine

        \effector({$(\msf{apply}, p, m)$}){
            $\mathit{elt} \gets$ unique element of $\mathit{elts}$ s.t.\ $\mathit{elt}.p = p$, or $\msf{null}$ if none exists (already deleted) \;
            \If{$\mathit{elt} \neq \msf{null}$}{
                $\mathit{elt}.\sigma \gets \mc{C}.\msf{eff}(m, \mathit{elt}.\sigma)$
            }
        }
    }
    \BlankLine
\caption{List of $\mc{C}$s as an operation-based CRDT.}
\label{alg:list_of_crdts}
\end{algorithm}

We sketch a proof of strong eventual consistency in \Cref{sec:correctness} (\Cref{thm:alg1}).
% TODO for PaPoC: cite arxiv version instead

% Long version: write about Unique Set, probably as subsection
% \begin{remark}
% In some applications, we need a collection of $\mc{C}$s, but the order of elements is not important. For example, consider a shared folder containing multiple collaborative documents. We define the \emph{unique set of $\mc{C}$s} 
% \end{remark}

\section{For-Each Operation}
\label{sec:for_each}
We now define our new CRDT operation, \emph{for-each}, on the list of $\mc{C}$s.

Let $O$ denote the set of all $\mc{C}$ operations. For technical reasons (described in \Cref{sec:pure} below), we restrict for-each to pure operations, where an operation is \emph{pure} if its generated message is just the operation itself. Formally:
\begin{definition}
Let $\mc{C}.\msf{gen}(o, \sigma)$ denote $\mc{C}$'s generator. An operation $o \in O$ is \emph{pure} if $\mc{C}.\msf{gen}(o, \sigma) = o$ for all states $\sigma$.
\end{definition}
Baquero et~al.\ \cite{pure_crdts} show that many classic CRDTs' operations are pure, at least with the relaxations discussed in \Cref{sec:pure}.

Let $O_P \subset O$ denote the subset of pure $\mc{C}$ operations. Let
\[
f: (p, \mathit{prior}) \ra O_P \cup \{\msf{del}, \msf{null}\}
\]
be a function that takes as input a list element's position $p$ and a boolean $\mathit{prior}$ described below, and returns one of:
\begin{itemize}
    \item $o \in O_P$: a pure $\mc{C}$ operation to apply to the element.
    \item $\msf{del}$: an instruction to delete the element.
    \item $\msf{null}$: an instruction to do nothing.
\end{itemize}

Then the operation $\msf{forEach}(f)$ loops over $\mathit{elts}$, applies $f$ to each element, then performs the operation specified by $f$. Specifically, it loops over all elements that are inserted causally prior or concurrently to the for-each operation itself, but not causally future elements. It also computes the argument $\mathit{prior}$ for $f$, which indicates whether each element is causally prior ($\msf{true}$) or concurrent ($\msf{false}$).

\begin{example}
\label{ex:rich_text_format}
In a rich text document, a user bolds a range of text. Let $\mathit{start}$ and $\mathit{end}$ be the positions of the first and last-plus-1 characters in the range, so that the range is $[\mathit{start}, \mathit{end})$. Define:

\kwfunction({$f(p, \mathit{prior})$}){
    \If{$\mathit{start} \le p < \mathit{end}$}{
        \Return{$(\mathit{rich} \mapsto \mathit{rich}.\mathit{attrs}.\msf{set}(\texttt{"bold"}, \msf{true}))$}
    }
    \lElse{\Return $\msf{null}$}
}
\BlankLine
\noindent
Then $\msf{forEach}(f)$ implements the intended behavior in \Cref{fig:rich_text_middle}: all characters in the range are bolded, including those inserted concurrently.
\end{example}

\begin{figure}
    \centering
    \includegraphics[width=0.45\textwidth]{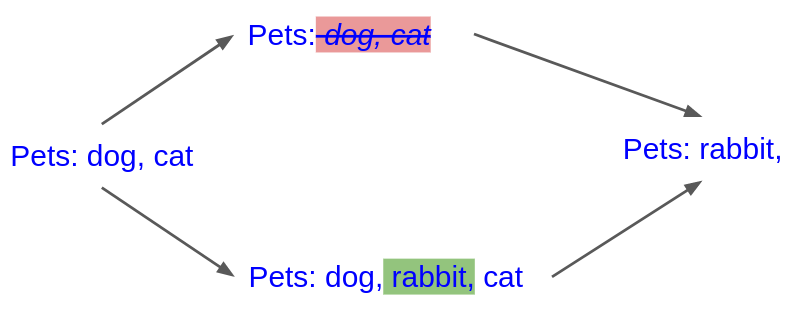}
    \caption{Typical user intention for a delete-range operation (top) concurrent to text insertion (bottom): the concurrent text is not deleted, to avoid data loss.}
    \label{fig:range_delete}
\end{figure}

\begin{example}
\label{ex:rich_text_delete}
Again in a rich-text document, a user deletes a range of text $[\mathit{start}, \mathit{end})$. To delete only existing characters, we consult $\mathit{prior}$:

\kwfunction({$f(p, \mathit{prior})$}){
    \If{$\mathit{prior}$ \textbf{and} $\mathit{start} \le p < \mathit{end}$}{
        \Return{$\msf{del}$}
    }
    \lElse{\Return $\msf{null}$}
}
\BlankLine
\noindent
Then $\msf{forEach}(f)$ implements the behavior shown in \Cref{fig:range_delete}. Note that we could instead use a $\msf{forEachPrior}$ operation, i.e., an ordinary loop on the initiating replica. However, $\msf{forEach}(f)$ generates less network traffic: a single $\msf{forEach}$ message for the entire range, instead of a separate $\msf{delete}$ message per deleted character.
\end{example}

\Cref{alg:for_each} gives a pseudocode implementation of for-each, which we now describe.

We first modify the list of $\mc{C}$s to track each element's logical insertion time $t$---namely, its sender's vector clock entry.\footnote{Some authors call this a \emph{causal dot}.} We also add a list $\mathit{buffer}$ to the internal state.

When a user calls $\msf{forEach}(f)$, their replica broadcasts $f$ together with the operation's vector clock $w$.

Upon receiving this message, a replica first loops over its current elements. For each element $\mathit{elt}$, the replica computes $f(\mathit{elt}.p, \mathit{prior})$ and does as instructed, but only locally; it does not broadcast any new messages. Here $\mathit{prior}$ indicates whether $\mathit{elt}$ was inserted causally prior to $\msf{forEach}(f)$. Note that we do not apply $f$ to elements that were already deleted on this replica, including by concurrent delete operations (the delete ``wins'').

Next, the receiving replica stores the message in its $\mathit{buffer}$. In the future, whenever the replica receives an insert message, it checks whether the insert operation is concurrent to $\msf{forEach}(f)$. If so, the replica computes $f(\mathit{elt}.p, \msf{false})$ and does as instructed, again only locally. Note that the message stays in the buffer forever, although in principle it could be discarded once all concurrent operations are received (i.e., it is causally stable \cite{pure_crdts}).

\begin{algorithm}[h]
    \cstate{
        $\mathit{elts}$: A list of elements $(p, \sigma, t)$, where $p$ is a position from $\mc{L}$, $\sigma$ is a state of $\mc{C}$, and $t = (\mathit{senderID}, \mathit{clock})$ is a vector clock entry; sorted by $p$ \;
        $\mathit{buffer}$: A list of pairs $(f, u)$, where $f$ is the $\msf{forEach}$ argument and $u$ is a vector clock entry \;
        $\mathit{vc}$: the local vector clock, in the form of a function from replica IDs to $\N$; initially the all-0 function \;
        $\mathit{replicaID}$: the unique ID of this replica
    }
    \BlankLine

    \kwfunction({$\msf{execute}(f, \mathit{elt}, \mathit{prior})$}){
        $\mathit{op} \gets f(\mathit{elt}.p, \mathit{prior})$ \;
        \If{$\mathit{op} \in O$}{
            $\mathit{elt}.\sigma \gets \mc{C}.\msf{eff}(\mathit{op}, \mathit{elt}.\sigma)$\label{line:effect}
        } \ElseIf{$\mathit{op} = \msf{del}$}{
            Delete $\mathit{elt}$ from $\mathit{elts}$
        }
    }
    \BlankLine
    
    \update({$\msf{insert}$}){
        \generator({$(i, \sigma)$}){
            $p \gets$ new $\mc{L}$ position between the positions at indices $i-1$ and $i$ in $\mathit{elts}$ \;
            $v \gets$ copy of $\mathit{vc}$ \;
            $v[\mathit{replicaID}] \gets v[\mathit{replicaID}] + 1$ \;
            \Return{$(\msf{insert}, p, \sigma, v, \mathit{replicaID})$}
        }
        \BlankLine

        \effector({$(\msf{insert}, p, \sigma, v, \mathit{senderID})$}){
            $\mathit{vc}[\mathit{senderID}] = v[\mathit{senderID}]$ \;
            $t \gets (\mathit{senderID}, v[\mathit{senderID}])$ \;
            Insert $(p, \sigma, t)$ into $\mathit{elts}$ \;
            // Loop over concurrent for-each operations. \;
            \For{$(f, u)$ \textbf{in} $\mathit{buffer}$}{\label{line:buffer_loop}
                $\mathit{concurrent} \gets (v[u.\mathit{senderID}] < u.\mathit{clock})$ \;
                \lIf{$\mathit{concurrent}$}{
                    $\msf{execute}(f, \mathit{elt}, \msf{false})$
                }
            }
        }
    }
    
    \update({$\msf{forEach}$}){
        \generator({$(f)$}){
            $w \gets$ copy of $\mathit{vc}$ \;
            $w[\mathit{replicaID}] \gets w[\mathit{replicaID}] + 1$ \;
            \Return{$(\msf{forEach}, f, w, \mathit{replicaID})$}
        }
        \BlankLine

        \effector({$(\msf{forEach}, f, w, \mathit{senderID})$}){
            $\mathit{vc}[\mathit{senderID}] = w[\mathit{senderID}]$ \;
            \For{$\mathit{elt}$ \textbf{in} $\mathit{elts}$}{\label{line:insert_loop}
                $\mathit{prior} \gets (w[\mathit{elt}.t.\mathit{senderID}] \ge \mathit{elt}.t.\mathit{clock})$ \;\label{line:prior_2}
                $\msf{execute}(f, \mathit{elt}, \mathit{prior})$
            }
            $u \gets (\mathit{senderID}, w[\mathit{senderID}])$ \;
            Append $(f, u)$ to $\mathit{buffer}$
        }
    }
    \BlankLine

\caption{List of $\mc{C}$s with our for-each operation. Blocks not shown here are the same as in \Cref{alg:list_of_crdts} ($\msf{elements}$, $\msf{delete}$, and $\msf{apply}$).}
\label{alg:for_each}
\end{algorithm}

% Technicality: our VCs only count insert & forEach ops, not delete/apply.

\subsection{On Pure Operations}
\label{sec:pure}
Our restriction to pure operations is not arbitrary: we need to know what message to pass to $\mc{C}.\msf{eff}$, even for concurrent elements. Such elements did not yet exist on the initiating replica, hence the replica could not pass their states to $\mc{C}.\msf{gen}$. With pure operations, we know that the generated message is just $o$ itself, as used on line~\ref{line:effect}.

In practice, we can relax the pure restriction by passing additional metadata to $\mc{C}.\msf{eff}$. In particular, we may pass in $f$'s vector clock: line~\ref{line:effect} of $\msf{execute}$ becomes
\[
\mathit{elt}.\sigma \gets \mc{C}.\msf{eff}((\mathit{op}, w), \mathit{elt}.\sigma).
\]
This does not threaten strong eventual consistency because $w$ is consistent across replicas.

$\mc{C}$ can use the provided vector clocks to query the causal order on operations. That is sufficient to implement most CRDTs using only pure operations \cite{pure_crdts}. List CRDTs' insert operations are a notable exception.

\subsection{Correctness}
Informally, we claim that \Cref{alg:for_each} matches the semantics described in the introduction. That is, a for-each operation's $f$ is applied to exactly the causally prior and concurrent elements, minus deleted elements, regardless of message order.

We defer a precise correctness claim and proof sketch to \Cref{sec:correctness} (\Cref{thm:alg2} and \Cref{cor:alg2}).
% TODO for PaPoC: cite arxiv version instead

\subsection{Other Data Structures}
For-each works equally well if we ignore the list order but still assign a unique ID $p$ to each element. That is, we can define a for-each operation on a \emph{set of CRDTs} in which each added element is assigned a unique ID.

Likewise, one can define for-each on a CRDT-valued map in which each key-value pair is assigned a unique ID when set, like Yjs's \texttt{Y.Map} shared type \cite{yjs}.

However, our construction does not work with a Riak-style map \cite{riak} in which a key's value CRDT is created on first use instead of explicitly set: two users may create the same key's value CRDT concurrently, complicating the choice of which for-each operations to apply \cite[\S 4]{semidirect}. We expect similar issues for the list of CRDTs in Kleppmann and Beresford's JSON CRDT \cite{Kleppmann2017json}, in which an element may reappear after deletion.

We leave full descriptions to future work.

\section{Examples}
\label{sec:examples}
We now describe example uses of our for-each operation in collaborative apps, at a high level. As in \Cref{sec:for_each}, we write a for-each operation as $\msf{forEach}(f)$, where
\[
f: (p, \mathit{prior}) \ra O_P \cup \{\msf{del}, \msf{null}\}
\]
is a function that takes as input a list element's position $p$ and whether it is causally prior (else concurrent), and outputs an instruction for that element: apply a (pure) operation $o \in O_P$, delete the element, or do nothing.

\subsection{Rich-Text Editor}
\label{sec:ex_rich_text}
Let us begin with a collaborative rich-text editor, as described in the introduction. To recap Examples~\ref{ex:rich_text_list} and \ref{ex:rich_text_format}, we can represent a rich-text document as a list of \emph{rich character CRDTs} $(\mathit{char}, \mathit{attrs})$, where $\mathit{char}$ is an immutable character and $\mathit{attrs}$ is a map CRDT for formatting attributes. Given list positions $\mathit{start}$ and $\mathit{end}$, define:

\kwfunction({$f(p, \mathit{prior})$}){
    \If{$\mathit{start} \le p < \mathit{end}$}{
        \Return{$(\mathit{rich} \mapsto \mathit{rich}.\mathit{attrs}.\msf{set}(\texttt{"bold"}, \msf{true}))$}
    }
    \lElse{\Return $\msf{null}$}
}
\BlankLine
\noindent
Then $\msf{forEach}(f)$ bolds the range $[\mathit{start}, \mathit{end})$ with the intended behavior in \Cref{fig:rich_text_middle}: all characters in the range are bolded, including concurrently-inserted ones.

It is possible to use a closed interval $[\mathit{start}, \mathit{end}']$ instead of the half-open interval $[\mathit{start}, \mathit{end})$. Here $\mathit{end}'$ is the position of the last character in the original range, while $\mathit{end}$ is the last-plus-one position. The difference is that $[\mathit{start}, \mathit{end})$ will also format concurrently-inserted characters at the end of the range, while $[\mathit{start}, \mathit{end}']$ will not. The latter behavior is typical for hyperlink formatting \cite{peritext}.

Other formatting attributes are similar. However, for deletions, one typically deletes only causally prior characters, as in \Cref{ex:rich_text_delete}. This is safer because deletions are monotonic (permanent), making unintended deletions harder to undo.

Note that a literal list of rich character CRDTs is memory-inefficient, since it stores a map CRDT per character. However, one can use this theoretical model as a guide, then implement an equivalent but more efficient CRDT. For example, one can store $\mathit{attrs}$'s state explicitly only when it differs from the previous character, like in Peritext \cite{peritext}.

% Long version: rich text inspired by (decompressed) Quill delta format, \url{https://quilljs.com/guides/designing-the-delta-format/}. Related work: Ritzy, rich-text CRDT that uses per-character formatting % https://doi.org/10.1145/1832772.1832777

\subsection{Recipe Editor}
A collaborative recipe editor allows multiple users to view and edit a recipe for a meal. Let us consider in particular the list of ingredients. We can model it as a list of \emph{ingredient CRDTs}, where each ingredient CRDT has sub-CRDTs for its name and amount.

% Long version: figure: ex recipe?

Suppose we add a ``scale recipe'' button that multiplies every amount by a given value. If one user scales the recipe, while concurrently, another user inserts a new ingredient, then it is important that the new ingredient's amount is also scaled. Otherwise, it will be out of proportion with the other ingredients.

To implement such a ``scale recipe'' operation, let $s$ be the scaling amount. Define $f$ by:

\kwfunction({$f(p, \mathit{prior})$}){
    \Return $(\mathit{ingredient} \ra \mathit{ingredient}.\mathit{amount}.\msf{mult}(s))$
}
\BlankLine
\noindent
Then $\msf{forEach}(f)$ scales every ingredient's amount, including ingredients inserted concurrently.\footnote{Here we assume an operation $\msf{mult}(s)$ on the ``ingredient amount'' CRDT. This is nontrivial if you also allow $\msf{set}(\mathit{value})$ operations, but it can be implemented using another list with for-each operations; we omit the details.
% Long version: include the details.
}

% https://docs.google.com/presentation/d/1VPv8wfHgEfAkpMPgmT6R_Y9jM8_Bb68GmXSKe-ounBY/edit#slide=id.gf516372b0d_0_28

\subsection{Slideshow Editor}
A slideshow editor is another collaborative app that can use for-each operations. A single slide might contain multiple images, shapes, or text boxes. These objects can be edited individually or together. For example, a user might translate (shift) a single object while another simultaneously rotates all objects on the slide. % These operations, though concurrent, modify different groups of objects on the slide, and should be applied to their respective objects.

To implement these translations and rotations, each translation on an object can be represented as a translation vector. Then the object's position is represented by a list CRDT $\mc{L}$ containing all translations made so far; the actual position is the sum of all the vectors in that list.

The entire slide can be represented as a list CRDT $\mc{L'}$, where each element is a list CRDT $\mc{L}_p$ of an object $p$'s translation vectors.\footnote{Since the objects on a slide are unordered, we use their positions $p$ merely as IDs, ignoring their total order.}

A user shifts an object $p$ by appending a translation vector to its list $\mc{L}_p$, and rotates an object $p$ by multiplying corresponding translation vectors in $\mc{L}_p$ by a rotation matrix. When a group of objects are edited, the same operation should be applied to each object's list. % When objects are edited concurrently, both translation vectors must be added.

% Long version: rotations around non-origin; why rotations commute (circle group)

\begin{figure}
    \centering
    \includegraphics[width=0.47\textwidth]{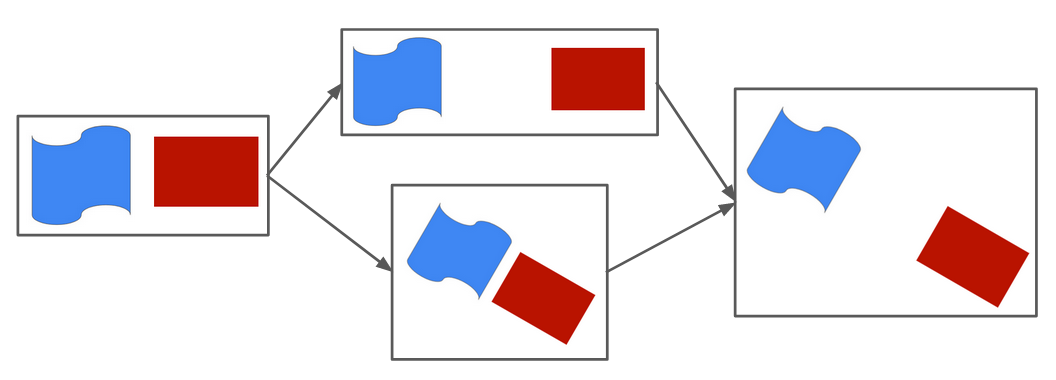}
    \caption{Typical user intention for a translation operation (top) concurrent to a group rotation (bottom): the rectangle is translated within the rotated group.}
    \label{fig:shapes}
\end{figure}

For example, say a user rotates a group of objects 30 degrees clockwise, like the bottom operation of \Cref{fig:shapes}. We want to rotate each object in the group. To keep objects aligned within the group, we also want to rotate any concurrent translation of those objects, such as the top translation of \Cref{fig:shapes}.

To implement this, let $\mathit{updatedObjects}$ be a set of the positions $p$ of the selected objects. Define $g$ and $f$ as:

\kwfunction({$g(q, \mathit{prior})$}){
    \Return{$\left(\mathit{vector} \mapsto \mathit{vector}.\msf{mult}\left( \begin{bmatrix} \cos 30 & \sin 30\\ -\sin 30 & \cos 30 \end{bmatrix}\right) \right)$}
}
\BlankLine
\noindent

\kwfunction({$f(p, \mathit{prior})$}){
    \If{$p \in \mathit{updatedObjects}$}{
        \Return{$(\mathit{object} \mapsto \mathit{object}.\mathit{forEach}(g))$}
    }
    \lElse{\Return $\msf{null}$}
}
\BlankLine
\noindent
Then $\msf{forEach}(f)$ rotates all $\mathit{updatedObjects}$ 30 degrees clockwise. 
% When each objects' position is computed, the vector modified by $g$ updates the object's position based on the rotation, as intended.\
The degree of rotation can also be stored to render the rotated object correctly.

\section{Related Work}
Dataflow programming and stream processing both perform operations ``for each'' element of a stream. Unlike this work, they typically apply a for-each operation to all regions in \Cref{fig:regions}, including the causal future. In particular, FlowPools \cite{flow_pools} allow issuing a for-each operation after a FlowPool (stream) begins; they apply the operation to all existing elements immediately, then store it as a callback for concurrent or future elements, similar to our algorithm (\Cref{alg:for_each}).

Operational Transformation \cite{ot_ressel} allows every operation on a collaborative app to perform a transformation for each concurrent operation. In contrast, we allow for-each operations to transform list elements (equivalently, insert operations) but not each other. Thus we do not need complicated algebraic rules to ensure eventual consistency.

The semidirect product of CRDTs \cite{semidirect} combines the operations of two CRDTs, $\mc{C}_1$ and $\mc{C}_2$, in a single CRDT. It essentially implements the rule: to apply a $\mc{C}_2$ operation, ``act on'' each prior and concurrent $\mc{C}_1$ operation in some way, then reduce over those $\mc{C}_1$ operations to get the current state. However, instead of storing the literal list of $\mc{C}_1$ operations, it only stores their reduced form (the actual state). Thus one can view the semidirect product as an optimized but less intuitive version of our list with for-each operations.

\section{Conclusions}
We formalized the list of CRDTs and described a novel for-each operation on this list. The resulting CRDT models a list of mutable values in a collaborative app, equipped with the operation: for each element of the list, including ones inserted concurrently, apply some operation to that element. We gave several examples in which our for-each operation matches user intention better than a literal for-each loop.

For future work, we plan to implement our for-each operation in the Collabs CRDT library \cite{collabs}.

\begin{acks}
We thank James Riely for insightful questions about a previous paper \cite{semidirect} that inspired this work. We also thank the anonymous PaPoC reviewers for helpful feedback. Matthew Weidner was supported by an NDSEG Fellowship sponsored by the US Office of Naval Research. Benito Geordie was supported by an REU sponsored by the US National Science Foundation.

% TODO for PaPoC: remove "PaPoC" from "reviewers"
\end{acks}

\bibliographystyle{ACM-Reference-Format}
\bibliography{references}

% TODO for PaPoC: comment out appendix
\appendix

\section{Correctness Proofs}
\label{sec:correctness}
\begin{theorem}
\label{thm:alg1}
The list of $\mc{C}$s (\Cref{alg:list_of_crdts}) satisfies strong eventual consistency.
\end{theorem}
\begin{proof}[Proof sketch]
It suffices to prove that concurrent operations commute \cite[Proposition 2.2]{crdt_survey_2011}.

Any operations that reference different positions trivially commute. In particular, insert operations commute with all operation types, since an insert operation always uses a new position. Delete operations on the same position trivially commute. Apply operations on the same position commute because $\mc{C}$ is a CRDT. Finally, a delete and apply operation on the same position commute because in either case, the element ends up deleted.
\end{proof}

\begin{theorem}[{Semantics of \Cref{alg:for_each}}]
\label{thm:alg2}
Fix an \Cref{alg:for_each} replica and a point in time. Let $\mathit{elt}$ be an element whose insert message has been received by the replica. If $\mathit{elt}$ is present in $\mathit{elts}$, then $\mathit{elt}.\sigma$ is the result of effecting the following messages on its initial state, exactly once and in causal order:
\begin{enumerate}[(a)]
    \item All messages due to apply operations on $\mathit{elt}$ that have been received.
    \item All (pure) operations $f(\mathit{elt}.p, \msf{true})$, where $\msf{forEach}(f)$ has been received and is causally future to $\mathit{elt}$'s insert operation.
    \item All (pure) operations $f(\mathit{elt}.p, \msf{false})$, where $\msf{forEach}(f)$ has been received and is concurrent to $\mathit{elt}$'s insert operation.
\end{enumerate}
The element is deleted (no longer present in $\mathit{elts}$) if and only if the replica has received a delete message for $\mathit{elt}$ or one of the above $f$ calls returned $\msf{del}$.
\end{theorem}
\begin{proof}[Proof sketch]
Type (a) messages are effected by apply's effector.

Type (b) operations are effected by for-each's effector when $\msf{forEach}(f)$ is received. Note that line~\ref{line:prior_2} correctly sets $\mathit{prior}$ to $\msf{true}$ by properties of vector clocks.

For type (c) operations, there are two cases. If $\mathit{elt}$'s insert message was received before the for-each message, then it is similar to the previous paragraph.

Otherwise, the operation is effected when the insert effector loops over $\mathit{buffer}$. Note that it correctly uses $\mathit{prior} = \msf{false}$. Since $\mathit{buffer}$ is in order by receipt time, the loop effects this operation in causal order relative to other type (c) messages. Also, since $\mathit{elt}$ is newly inserted, the operation is effected prior to all type (a) and (b) messages; this respects the causal order because all such messages are concurrent or causally future.

The claim about deletions is similar.
\end{proof}

\begin{corollary}
\label{cor:alg2}
\Cref{alg:for_each} satisfies strong eventual consistency.
\end{corollary}
\begin{proof}
By the theorem's description and the fact that $\mc{C}$ is a CRDT.
\end{proof}

\balance

\end{document}